\documentclass[letterpaper, 10pt, conference]{ieeeconf}
\usepackage[pdftex]{graphicx}
\usepackage{amsmath,amsfonts,amssymb}
\usepackage{mathtools}
\usepackage{times}
\usepackage{float}
\usepackage{multirow}
\usepackage{paralist, tabularx}
\usepackage{hyperref}
\usepackage{algorithm}
\usepackage{algpseudocode}
\usepackage[normalem]{ulem}
\usepackage{cite}
\usepackage{color}
\usepackage{booktabs}
\usepackage{makecell}
\usepackage{epsfig}
\usepackage{siunitx}
\usepackage{colortbl}
\usepackage[dvipsnames]{xcolor}
\usepackage{tikz}
\usetikzlibrary{arrows, positioning, calc}
\usepackage{soul}

\newtheorem{theorem}{Theorem}
\newtheorem{remark}{Remark}
\newtheorem{lemma}{Lemma}
\newtheorem{proposition}{Proposition}
\newtheorem{corollary}{Corollary}
\newtheorem{definition}{Definition}
\newtheorem{assumption}{Assumption}

\newcommand{\R}{\mathbb{R}}
\newcommand{\C}{\mathbb{C}}
\newcommand{\N}{\mathbb{N}}
\newcommand{\norm}[1]{\left\|#1\right\|}

\newcommand{\diag}{\operatorname{diag}}
\newcommand{\spec}{\operatorname{spec}}
\newcommand{\Int}{\operatorname{Int}}
\newcommand{\Real}{\operatorname{Re}}
\newcommand{\Imag}{\operatorname{Im}}
\newcommand{\adj}{\operatorname{adj}}

\newcommand{\rank}{\operatorname{rank}}
\newcommand{\Ac}{\widetilde{A}}
\newcommand{\bc}{\widetilde{b}}
\newcommand{\Ld}{\Lambda^{\dagger}}
\DeclareMathOperator*{\argmin}{arg\,min}

\providecommand{\IEEEoverridecommandlockouts}{}
\providecommand{\overrideIEEEmargins}{}
\providecommand{\QED}{\hfill\blacksquare}

\makeatletter
\@ifundefined{iflatexml}{%
  \newif\iflatexml
  \latexmlfalse
}{}
\makeatother

\IEEEoverridecommandlockouts
\title{\LARGE \bf
Safe Stabilising Full-Order Affine Control Barrier Functions for Linear Systems (Extended)
}

\author{
Faisal Lawan, Joaquin Carrasco, Lanlan Su
\thanks{This work was supported by the University of Manchester's School of Engineering Strategic PGR Scholarship.}
\thanks{F. Lawan, J. Carrasco, and L. Su are with the Department of Electrical and Electronic Engineering, The University of Manchester, M13 9PL, United Kingdom (Corresponding authors: F. Lawan, \texttt{faisal.lawan@manchester.ac.uk}; J. Carrasco, \texttt{joaquin.carrasco@manchester.ac.uk}; L. Su, \texttt{lanlan.su@manchester.ac.uk}.)}
}

\begin{document}
\maketitle
\thispagestyle{empty}
\pagestyle{empty}

\begin{abstract}
Control barrier function safety filters enforce constraints by modifying a nominal input, but the resulting switching can destabilise the closed loop even when the nominal and filtered modes are individually stable. This paper presents a design framework for safe and globally exponentially stabilising controllers for linear systems with a single full-relative-degree affine constraint. We show that the filtered-mode spectrum is fixed by the barrier tuning and is independent of the plant,
nominal controller, and quadratic-program weighting. This structure yields an explicit nominal controller for which the safety filter remains inactive everywhere. For a prescribed nominal controller, we prove that the nominal and filtered modes admit a strong common quadratic Lyapunov function if and only if the ratio of the nominal characteristic polynomial to the barrier polynomial is strongly strictly positive real. This equivalence characterises the existence of a common quadratic Lyapunov and provides a scalar frequency-domain test, along with an explicit interval of admissible gains. Building on these results, the extended analysis derives an explicit common storage function and reduces an existing LMI synthesis condition to a feasibility test in a single matrix variable. A flexible two-mass example explains a known instability mechanism and demonstrates how the proposed design restores safety and global exponential stability.
\end{abstract}

\section{Introduction}

Safety-critical control requires a system to satisfy constraints while
pursuing its nominal control objective. Control barrier functions
(CBFs) provide a systematic way to enforce such constraints through
forward invariance \cite{ames2017cbf,ames2019cbf}. They are commonly
implemented as safety filters: at each state, a quadratic program
modifies a nominal input only when necessary to preserve safety.
High-order CBFs extend this construction to constraints for which the
input does not appear in the first derivative of the constraint
function \cite{xiao2022hocbf}. These methods can be placed around an
existing controller, but the resulting closed-loop behaviour need not
retain the stability properties of that controller.

This limitation has motivated a growing body of work on the dynamics
induced by barrier-function safety filters. Such filters can introduce
undesired equilibria \cite{reis2021undesirable,tan2024undesired},
unbounded trajectories \cite{mestres2026dynamical}, and limit cycles
\cite{mestres2025characterization}. For high-order barrier functions,
Marchese \emph{et al.} \cite{marchese2025hocbf} exhibited a particularly
important failure mechanism: in a flexible two-mass system, the
nominal and filtered dynamics are both stable, yet the complete closed
loop diverges because it repeatedly switches between them. Stability
of the individual modes is therefore insufficient to guarantee
stability of the safety-filtered system.

The joint design of safe and stabilising feedback has been studied
using modified quadratic programs, compatibility conditions, and
optimisation-based constructions with certified regions of attraction
\cite{jankovic2018robust,cortez2022compatibility,mestres2023optimization}.
Most recently, Mestres \emph{et al.}
\cite{mestres2026safefeedback} characterised linear state-feedback
gains that globally satisfy high-order barrier constraints while
stabilising linear systems. These approaches synthesise feedback laws
that simultaneously meet safety and stability requirements. The
complementary question considered here is whether stability is
preserved when a minimally invasive safety filter is wrapped around a
nominal controller, thereby creating state-dependent switching.

For linear systems subject to a single affine constraint, Mestres
\emph{et al.} \cite{mestres2026dynamical} showed that the
quadratic-program safety filter produces a state-dependent piecewise
affine system with a nominal mode and a filtered mode. Their analysis
relates equilibria and unbounded trajectories to the eigenstructure of
the filtered mode and gives sufficient conditions for global
exponential stability in terms of invariant zeros and a common
quadratic Lyapunov function. It also provides an LMI construction for
a stabilising nominal controller. These results apply to general
relative degree, but they do not expose the special structure that
arises when the constraint has relative degree equal to the system
order. Consequently, they do not directly characterise which nominal
controllers admit a common quadratic certificate in this case.

The full-relative-degree case deserves separate treatment because the
constrained output and its derivatives describe the complete state,
leaving no internal dynamics. It includes chains of integrators with a
constraint on the terminal output \cite{Abel2022}, as well as the
flexible two-mass system of \cite{marchese2025hocbf}. In this setting,
the invariant-zero requirement of
\cite[Cor.~1]{mestres2026dynamical} is automatically satisfied, and
the remaining stability question concerns the interaction between the
nominal and filtered modes.

The conference version of this work established the filtered-mode
structure, the inactive-filter design, and the positive-real
characterisation. The present extended version builds on those results
by developing their passivity interpretation, deriving an explicit
common storage function, and analysing feasibility of the existing
LMI synthesis. Its contributions are:
\begin{enumerate}
    \item We show that the filtered-mode spectrum is fixed by the
    barrier tuning and is independent of the plant, nominal
    controller, and quadratic-program weighting. Full relative degree
    also forces the input matrix to have rank one.

    \item We derive a nominal controller that keeps the safety filter
    inactive while guaranteeing safety and global exponential
    stability.

    \item We prove that the nominal and filtered modes admit an SCQLF
    if and only if the ratio of the nominal characteristic polynomial
    to the barrier polynomial is SSPR. The associated storage function
    is also an explicit common quadratic Lyapunov function.

    \item We obtain an explicit interval of admissible nominal gains
    and a reduced condition in one matrix variable that determines
    feasibility of the LMI synthesis in
    \cite{mestres2026dynamical}.

    \item We explain and resolve the instability observed in the
    flexible two-mass example of \cite{marchese2025hocbf}.
\end{enumerate}
\section{Preliminaries}

\subsection{Notation}

We write $\N,\R,\R_{\geq0},\C$ for the natural, real, non-negative real and complex numbers, respectively, $0_n$ and $I_n$ for the $n$-dimensional zero vector and the $n \times n$ identity, respectively, $e_i$ for the $i$-th canonical basis vector of $\R^n$, and $[n]:=\{1,\dots,n\}$. For $S\subset\R^n$, $\Int(S)$ and $\partial S$ are its interior and boundary. For a square symmetric matrix $A=A^\top$, $A\succ0$ and $A\prec0$ mean that $A$ is positive and negative definite, respectively. For $G\succ0$ and $x\in\R^n$, define $\norm{x}_G=\sqrt{x^\top Gx}$. We write $\det$, $\rank$, and $\spec$ for determinant, rank, and spectrum, $\Real$ and $\Imag$ for real and imaginary parts, and $\circledast$ for discrete convolution. Define $p_M(s):=\det(sI_n-M)$ as the characteristic polynomial of a square matrix $M$, and $[-\infty,\infty]:=\R\cup\{-\infty,+\infty\}$.

\subsection{Common Quadratic Lyapunov Functions}

\begin{definition}\label{def:cqlf}
    Hurwitz matrices $A_1,A_2\in\R^{n\times n}$ are said to admit a \emph{strong common quadratic Lyapunov function} (SCQLF) if there exists a $P=P^\top\succ0$ with $A_i^\top P+PA_i\prec0$ for $i=1,2$.
\end{definition}

\begin{lemma}[{\cite[Thm.~1]{king2006cqlf}}]\label{lem:king}
    Let $A_1,A_2\in\R^{n\times n}$ be Hurwitz with $\rank(A_1-A_2)=1$. Then $A_1$ and $A_2$ admit an SCQLF if and only if $A_1A_2$ has no negative real eigenvalue.
\end{lemma}

\begin{lemma}[{\cite{shorten2004unifying}}]\label{lem:circle}
    Let $A\in\R^{n\times n}$ be Hurwitz, with $b,c\in\R^n$ such that $(A,b,c^\top)$ is minimal, and let $\lambda\in\R$. Then $A$ and $A-\lambda bc^\top$ admit an SCQLF if and only if
    \begin{equation}\label{eq:circle}
      1+\lambda\Real\!\left\{c^\top(j\omega I_n-A)^{-1}b\right\}>0
      \quad\text{for all }\omega\in\R .
    \end{equation}
\end{lemma}

The following lemma shows directly that the positive-real condition implies that the rank-one perturbation is Hurwitz.

\begin{lemma}\label{lem:hurwitz}
    Let $A\in\R^{n\times n}$ be Hurwitz and let $b,c\in\R^n$ satisfy $\rank(bc^\top)=1$. If $1+\Real\{c^\top(j\omega I_n-A)^{-1}b\}>0$ for all $\omega\in\R$, then $A-bc^\top$ is Hurwitz.
\end{lemma}

\begin{proof}
Define \[ F(s):=1+c^\top(sI_n-A)^{-1}b. \] Since $A$ is Hurwitz, $sI_n-A$ is nonsingular for $\Real\{s\}\geq0$. Thus, by the matrix determinant lemma, \[ F(s) =\frac{\det\!\left(sI_n-(A-bc^\top)\right)} {\det(sI_n-A)}, \qquad \Real\{s\}\geq0. \] The hypothesis gives $\Real \{F(j\omega)\}>0$ for all $\omega\in\R$, so the Nyquist plot of $F$ lies in the open right half-plane and hence does not encircle the origin. Since $F$ has no poles in the closed right half-plane, the Nyquist criterion implies that $F$ has no zeros there. Therefore, \[ \det\!\left(sI_n-(A-bc^\top)\right)\neq0, \qquad \Real\{s\}\geq0, \] and hence $A-bc^\top$ is Hurwitz.
\end{proof}

\subsection{Positive Realness and Passivity}

Following \cite{brogliato2020dissipative}, a system with input $u$ and output $y$ where $u(t),y(t)\in \R^m$ is \emph{very strictly passive} (VSP) if there exist $\beta\in\R$, and $\delta>0$, $\varepsilon>0$ such that
\begin{equation}\label{eq:vsp}
\begin{aligned}
  \int_0^t\! y(\tau)^\top u(\tau)\,d\tau
  \geq{}& \beta+\delta\!\int_0^t\! u(\tau)^\top u(\tau)\,d\tau\\
  &+\varepsilon\!\int_0^t\! y(\tau)^\top y(\tau)\,d\tau,
\end{aligned}
\end{equation}
for all admissible $u$ and all $t\geq0$.

\begin{definition}[{\cite[Def.~2.78]{brogliato2020dissipative}}]\label{def:sspr}
    A SISO real-rational transfer function $H: \mathbb{C} \to \mathbb{C}$ is \emph{strongly strictly positive real} (SSPR) if $H$ is analytic in $\Real\{s\}\geq0$ and $\Real\{H(j\omega)\} \geq \delta>0$ for all $\omega\in[-\infty,\infty]$ and some $\delta>0$.
\end{definition}

\begin{lemma}[{\cite[Thm.~2.81]{brogliato2020dissipative}}]\label{lem:sspr-vsp}
    Let $(A,b,c,d)$ be a minimal realisation of the transfer function $H$. Then $H$ is SSPR if and only if the associated system is VSP.
\end{lemma}

Consider the transfer function
\[
    g(s)=c^\top(sI-A)^{-1}b.
\]
Condition~\eqref{eq:circle} can therefore be interpreted as an SSPR condition on the return difference $1+\lambda g$.

\subsection{Safety Filters for Linear Systems and Affine Constraints}
\label{subsec:filter-setup}

Consider a continuous-time linear system 
\[
    \dot x = Ax+Bu,
\]
with $x\in\R^n$, $u\in\R^m$, $A\in\R^{n\times n}$, and $B\in\R^{n\times m}$. Let 
\[
    h_0(x):=c^\top x+d,
\]
where $c\in\R^n$ and $d\in\R$, and define the corresponding constraint set 
\[
    \mathcal{C}_0:=\{x:h_0(x)\geq0\}.
\]

Suppose $h_0$ has relative degree $r\in[n]$ on $\R^n$
\cite[Def.~3]{xiao2022hocbf}, and define
$h_i=\dot h_{i-1}+\alpha_i h_{i-1}$ and
$\mathcal{C}_i=\{x:h_i(x)\geq0\}$, where $\alpha_i>0$, for
$i\in[r-1]$. By \cite{xiao2022hocbf}, any locally Lipschitz controller satisfying
\begin{equation}\label{eq:forward-invariance}
(\nabla h_{r-1}(x))^\top(Ax+Bu)+\alpha_r h_{r-1}(x)\geq0
\end{equation}
renders $\mathcal{C}:=\cap_{i=0}^{r-1}\mathcal{C}_i$ forward invariant.

\begin{equation}\label{eq:alpha-def}
\varphi(s):=\prod_{j=1}^{r}(s+\alpha_j),\qquad
\alpha:=\prod_{i=1}^{r}\alpha_i .
\end{equation}
Using $c^\top A^iB=0_m^\top$ for $i\leq r-2$, condition~\eqref{eq:forward-invariance} becomes \cite[Lem.~1]{mestres2026dynamical}
\begin{equation}\label{eq:hocbf-linear}
  c^\top\varphi(A)x + c^\top A^{r-1}Bu + \alpha d \geq 0 .
\end{equation}

Given a nominal controller $k(x)=-Kx$ and a weight $G\succ0$, the safety filter is
\begin{equation}\label{eq:filter}
  u^*(x)=\argmin_{u\in\R^m}\tfrac12\norm{u-k(x)}_G^2
  \quad\text{s.t. }\eqref{eq:hocbf-linear},
\end{equation}
illustrated in Fig.~\ref{fig:interconnection}(a).

\begin{assumption}\label{as:origin}
For the safe set $\mathcal C_0$ defined in
Section~\ref{subsec:filter-setup}, $0_n\in\Int(\mathcal C_0)$, or
equivalently, $d>0$.
\end{assumption}

\begin{assumption}\label{as:reldeg}
The affine constraint $h_0$ introduced in
Section~\ref{subsec:filter-setup} has relative degree $r=n$ on
$\R^n$.
\end{assumption}

Assumption~\ref{as:origin} is that of \cite[Sec.~II]{mestres2026dynamical}, which also
assumes $(A,B)$ stabilisable; Lemma~\ref{lem:minimal} shows that
Assumption~\ref{as:reldeg} implies that $(A,B)$ is stabilisable.

Following \cite[Sec.~II]{mestres2026dynamical}, define
\begin{align}
  \theta &:= \norm{B^\top(A^\top)^{r-1}c}_{G^{-1}}^2, &
  \Lambda &:= c^\top A^{r-1}B,\label{eq:lambda} \\
  v_1 &:= -\frac{BG^{-1}B^\top(A^\top)^{r-1}c}{\theta},\label{eq:v1} &
  \Ld &:= \frac{G^{-1}\Lambda^\top}{\theta},
\end{align}
so that $\Lambda\Ld=1$ and $v_1=-B\Ld$, together with
\begin{equation}\label{eq:v2}
  v_3^\top := c^\top\varphi(A),\quad
  v_2^\top := v_3^\top-\Lambda K,\quad
  \eta(x) := v_2^\top x+\alpha d,
\end{equation}
and set $R_+:=\{x:\eta(x)\geq0\}$, $R_-:=\R^n\setminus R_+$. Relative degree $r$ gives $\Lambda\neq0_m^\top$, so \eqref{eq:hocbf-linear} is feasible everywhere and \eqref{eq:filter} has the closed form \cite[Sec.~II]{mestres2026dynamical}
\begin{equation}\label{eq:ustar}
  u^*(x)=
  \begin{cases}
    -Kx, & x\in R_+,\\[2pt]
    -Kx-\dfrac{\eta(x)}{\theta}G^{-1}B^\top(A^\top)^{r-1}c, & x\in R_-,
  \end{cases}
\end{equation}
giving the piecewise affine closed loop
\begin{equation}\label{eq:closedloop}
  \dot x=
  \begin{cases}
    A_0x, & x\in R_+,\\
    \Ac x+\bc, & x\in R_-,
  \end{cases}
  \qquad
  \begin{aligned}
    A_0 &:= A-BK,\\
    \Ac &:= A_0+v_1v_2^\top,
  \end{aligned}
\end{equation}
with $\bc:=\alpha dv_1$. By \cite[Thm.~2]{xu2015robustness} and
\cite{mestres2025regularity}, $u^*$ is locally Lipschitz, so solutions are unique and $\mathcal{C}$ is forward invariant. We call $A_0$ the nominal mode and $\Ac$ the filtered mode.

\section{Safety Filters of Full Relative Degree}

This section identifies the structural consequences of full relative
degree. We first show that the constrained input--output channel
captures the complete state and forces a rank-one input matrix. We then
use the resulting coordinates to characterise the filtered mode and
construct a nominal controller that eliminates switching.

Throughout the subsequent sections, Assumptions~\ref{as:origin}--\ref{as:reldeg} hold. Let $\mathcal{O}\in\R^{n\times n}$ be the observability matrix of $(A,c^\top)$, whose $i$-th row is $c^\top A^{i-1}$.

\subsection{Minimality and Rank of the Input Matrix}

Full relative degree removes the internal zero dynamics that otherwise
enter the stability conditions of
\cite{mestres2026dynamical}. The following lemma makes this property
explicit and supplies the input--output identities used throughout the
subsequent analysis.

\begin{lemma}\label{lem:minimal}
Suppose Assumption~\ref{as:reldeg} holds, and let $\Lambda$ and $A_0$
be defined by \eqref{eq:lambda} and \eqref{eq:closedloop},
respectively.
The matrix $\mathcal{O}$ is invertible, the pair $(A,B)$ is controllable, and
\begin{equation}\label{eq:OB}
  \mathcal{O}B=e_n\Lambda,
  \quad\text{so}\quad
  B=\left(\mathcal{O}^{-1}e_n\right)\!\Lambda,\ \ \rank(B)=1 .
\end{equation}
Moreover, for every $K\in\R^{m\times n}$ and
$s\notin\spec(A_0)$,
\[
c^\top(sI_n-A_0)^{-1}B=\frac{\Lambda}{p_{A_0}(s)}.
\]
\end{lemma}

\begin{proof}
Relative degree $n$ gives $c^\top A^iB=0_m^\top$ for $i\leq n-2$ and $c^\top A^{n-1}B=\Lambda\neq0_m^\top$. Pick $w\in\R^m$ with $\Lambda w\neq0$ and set $b:=Bw$, so that $c^\top A^ib=0$ for $i\leq n-2$ and $c^\top A^{n-1}b=\Lambda w\neq0$. Let $\mathcal{A}:=\begin{bmatrix}b & Ab & \cdots & A^{n-1}b\end{bmatrix}$ be the controllability matrix of $(A,b)$. The $(i,j)$ entry of $\mathcal{O}\mathcal{A}$ is $c^\top A^{i+j-2}b$, which vanishes for $i+j\leq n$ and equals $\Lambda w$ for $i+j=n+1$. The matrix $\mathcal{O}\mathcal{A}$ therefore vanishes above its anti-diagonal and equals $\Lambda w$ along it, so $\det(\mathcal{O}\mathcal{A})=(-1)^{n(n-1)/2}(\Lambda w)^n\neq0$. Since both factors are square, $\mathcal O$ and $\mathcal A$ are invertible. Thus, $(A,c^\top)$ is observable and $(A,b)$ is controllable, so $(A,b,c^\top)$ is minimal. Since $b$ lies in the range of $B$, the pair $(A,B)$ is controllable as well.

The rows of $\mathcal{O}B$ are $c^\top A^{i-1}B$, which vanish for $i\leq n-1$ and equal $\Lambda$ for $i=n$, giving \eqref{eq:OB}.

For the last claim, $c^\top A^iB=0$ for $i\leq n-2$ and $A_0 = A-BK$ gives $c^\top A_0^i=c^\top A^i$ for all $i\leq n-1$. Thus, the Markov parameters of $(A_0,B,c^\top)$ agree with those of $(A,B,c^\top)$ up to order $n-1$ and $c^\top(sI-A_0)^{-1}B=\Lambda s^{-n}+O(s^{-n-1})$. Writing $(sI-A_0)^{-1}=\adj(sI-A_0)/p_{A_0}(s)$, the numerator
$N(s):=c^\top\adj(sI-A_0)B$ has degree at most $n-1$, while
$N(s)=p_{A_0}(s)\left(\Lambda s^{-n}+O(s^{-n-1})\right)\to\Lambda$ as $s\to\infty$, because $p_{A_0}$ is a monic polynomial of degree $n$. A polynomial of degree at most $n-1$ with a finite limit at infinity is constant, so $N\equiv\Lambda$.
\end{proof}

Since $B=B\Ld\Lambda$ by \eqref{eq:OB} and $\Lambda\Ld=1$,
\begin{equation}\label{eq:Bproj}
  B\left(I_m-\Ld\Lambda\right)=0 .
\end{equation}
Lemma~\ref{lem:minimal} also shows that the auxiliary system
$(A_0,BG^{-1}B^\top(A^\top)^{n-1}c,c^\top)$ of \cite[Cor.~1]{mestres2026dynamical} has transfer function $\theta/p_{A_0}(s)$ and hence no finite invariant zeros. Its invariant zero condition is therefore met automatically, and under Assumption~\ref{as:reldeg} the
corollary simplifies to its remaining requirement, that $A_0\Ac$ have no negative real eigenvalue, which is also the second hypothesis of \cite[Thm.~2]{mestres2026dynamical}.

\subsection{The Filtered Mode}

The preceding lemma allows the constrained output and its derivatives
to be used as coordinates for the full state. In these coordinates,
the active-filter dynamics contain no internal modes and are determined
directly by the barrier polynomial.

Following \cite{isidori1995nonlinear}, define the transverse variable
\begin{equation}\label{eq:zeta}
  \zeta := \begin{bmatrix} h_0(x) & \dot h_0(x) & \cdots & h_0^{(n-1)}(x)\end{bmatrix}^\top
  = \mathcal{O}x+d\, e_1,
\end{equation}
the second equality following from $h_0(x)=c^\top x+d$ and $h_0^{(i)}(x)=c^\top A^ix$. Write $\varphi(s)=s^n+\sum_{i=1}^{n}\beta_is^{n-i}$, so that $\beta_n=\alpha$, and let
\begin{equation}\label{eq:Gamma}
  \Gamma := \begin{bmatrix} 0_{n-1} & I_{n-1}\\ -\beta_n & -\beta_{n-1}\ \cdots\ -\beta_1
  \end{bmatrix}
\end{equation}
be the companion matrix of $\varphi$, so $p_\Gamma=\varphi$.

\begin{proposition}\label{prop:Atilde}
Suppose Assumptions~\ref{as:origin} and \ref{as:reldeg} hold. Let
$\varphi$ be defined by \eqref{eq:alpha-def}, let $\mathcal O$ be the
observability matrix defined above, and let $\Gamma$ be defined by
\eqref{eq:Gamma}. For every $K\in\R^{m\times n}$ and $G\succ0$, let
$\Ac$ and $\bc$ be defined by \eqref{eq:closedloop}. Then
\begin{equation}\label{eq:Atilde-form}
  \Ac=\mathcal{O}^{-1}\Gamma\mathcal{O}=A-B\Ld c^\top\varphi(A),
  \qquad \bc=d\,\mathcal{O}^{-1}\Gamma e_1 .
\end{equation}
Hence $\spec(\Ac)=\{-\alpha_i\}_{i=1}^n$, so $\Ac$ is Hurwitz and invertible, and
$\Ac$ depends on neither $K$ nor $G$.
\end{proposition}

\begin{proof}
On $R_-$ the constraint \eqref{eq:hocbf-linear} holds with equality, so
$\dot h_{n-1}=-\alpha_n h_{n-1}$ along solutions; with
$h_i=\dot h_{i-1}+\alpha_i h_{i-1}$ this gives $\varphi(\tfrac{d}{dt})h_0=0$, which in the coordinates \eqref{eq:zeta} reads $\dot\zeta=\Gamma\zeta$. Differentiating \eqref{eq:zeta} gives $\mathcal{O}\dot x=\Gamma(\mathcal{O}x+d e_1)$, and $\mathcal{O}$ is invertible by Lemma~\ref{lem:minimal}, which yields the first equality in \eqref{eq:Atilde-form} and the expression for $\bc$. For the second equality, use $v_1=-B\Ld$ and $v_2^\top=v_3^\top-\Lambda K$ in \eqref{eq:closedloop}:
\begin{align*}
  \Ac &= A-BK-B\Ld\!\left(v_3^\top-\Lambda K\right)\\
      &= A-B\Ld v_3^\top-B\!\left(I_m-\Ld\Lambda\right)\!K,
\end{align*}
and the last term vanishes by \eqref{eq:Bproj}. Similarity to $\Gamma$ gives $\spec(\Ac)=\spec(\Gamma)=\{-\alpha_i\}\subset\R_{<0}$.
\end{proof}

Since $v_1=-B\Ld$ and $v_3^\top=c^\top\varphi(A)$, the second expression in \eqref{eq:Atilde-form} also reads
\begin{equation}\label{eq:Atilde-rank1}
  \Ac = A+v_1v_3^\top,
  \qquad\text{equivalently}\qquad
  \Ac-v_1v_3^\top = A .
\end{equation}
Two further identities follow from the construction: by \eqref{eq:OB} and relative degree $n$,
\begin{equation}\label{eq:Ov1}
  \mathcal{O}v_1=- e_n, \qquad \Gamma e_1=-\alpha e_n,
\end{equation}
so $\bc=-d\alpha\mathcal{O}^{-1} e_n=\alpha dv_1$ as in \eqref{eq:closedloop}.

For $r<n$ the class-$\mathcal{K}$ slopes account for only $r$ of the eigenvalues of $\Ac$. By \cite[Lem.~2]{mestres2026dynamical} the remainder are eigenvalues of $A_0$ or roots of $c^\top(\lambda I-A_0)^{-1}BG^{-1}B^\top(A^\top)^{r-1}c=0$, and in either case
depend on the plant and on $K$. When $r=n$ the slopes account for all $n$ eigenvalues, and activating the filter fixes the entire vector field.

By Proposition~\ref{prop:Atilde}, the matrix $\Ac$ has no positive real eigenvalue, so the sufficient condition for unbounded trajectories in \cite[Prop.~3]{mestres2026dynamical} is never met under Assumption~\ref{as:reldeg}. That condition is sufficient and not necessary, so unbounded trajectories may still arise by another mechanism, and Section~\ref{sec:example} exhibits one.

Proposition~\ref{prop:Atilde} also supplies the first hypothesis of \cite[Thm.~2]{mestres2026dynamical} at no cost, since $\Ac$ is Hurwitz for every $K$ and every $G$, so that theorem simplifies to its second hypothesis as well. The example of \cite[Sec.~III.B]{marchese2025hocbf} has $r=n$, hence $\Ac$ Hurwitz, and is reported in \cite[Sec.~V]{mestres2026dynamical} not to be globally exponentially stable. With both modes Hurwitz whenever $K$ is stabilising, any such failure must come from their interaction, which is the behaviour attributed to persistent switching in \cite{marchese2025hocbf}.

\subsection{Safety by Design}

Although Proposition~\ref{prop:Atilde} makes the filtered mode Hurwitz,
stability of the complete closed loop still depends on its interaction
with the nominal mode. The first design removes this interaction
entirely by making the two modes coincide.

\begin{proposition}\label{prop:deactivate}
Suppose Assumptions~\ref{as:origin} and \ref{as:reldeg} hold, and let
$\Ld$ and $v_3$ be defined by \eqref{eq:lambda}--\eqref{eq:v2}. Set
$K^\star:=\Ld v_3^\top$ and $k(x):=-K^\star x$. Then the safety filter
\eqref{eq:filter} is inactive for every $x\in\R^n$, i.e.,
$u^*(x)=k(x)=-K^\star x$ for every $x\in\R^n$. Consequently,
\eqref{eq:closedloop} reduces to $\dot x=A_0x$ with $A_0=\Ac$ and
$\spec(A_0)=\{-\alpha_i\}_{i=1}^n$, the safe set $\mathcal C$ defined
in Section~\ref{subsec:filter-setup} is forward invariant, and the
origin is globally exponentially stable.
\end{proposition}

\begin{proof}
By \eqref{eq:v2} and $\Lambda\Ld=1$, $v_2^\top=v_3^\top-\Lambda\Ld v_3^\top=0_n^\top$, so $\eta\equiv\alpha d$, which is positive by Assumption~\ref{as:origin}. Then $R_+=\R^n$ and $\Ac=A_0$, and Proposition~\ref{prop:Atilde} gives the spectrum.
\end{proof}

By \eqref{eq:Atilde-form}, $\Ac=A-BK^\star$, so $K^\star$ is also the gain the filter applies whenever it is active: the two modes of \eqref{eq:closedloop} are $A-BK$ and $A-BK^\star$, and the switching studied in \cite{marchese2025hocbf} is between the designer's gain and the input-output linearising one. Choosing $K=K^\star$ removes the
switching, which is not available for $r<n$ because the zero dynamics are then left untouched. The closed-loop poles are the negatives of the class-$\mathcal K$ slopes. If $\alpha_i=\rho\bar\alpha_i$, then
\[
K^\star(\rho)=\rho^n\!\left(\prod_{i=1}^n\bar\alpha_i\right)\Ld c^\top
+O(\rho^{n-1}),
\]
so $\norm{K^\star(\rho)}=\Theta(\rho^n)$. Thus aggressive barrier
tuning can require aggressive nominal control. Section~\ref{sec:frequency}
quantifies how far one may move away from this choice.

\section{A Positive-Real Characterisation}\label{sec:frequency}

Since Proposition~\ref{prop:Atilde} fixes the filtered mode, we treat $\Ac$ as the reference matrix and the nominal mode
\begin{equation}\label{eq:A0-rank1}
  A_0=\Ac-v_1v_2^\top
\end{equation}
as a rank-one perturbation of it, and pose the common quadratic Lyapunov function question in the frequency domain.

\subsection{The Auxiliary Feedback System and its Return Difference}

By \eqref{eq:Atilde-rank1}, the plant matrix is $A=\Ac-v_1v_3^\top$, so both $A$ and every nominal mode $A_0=\Ac-v_1v_2^\top$ are rank-one perturbations of $\Ac$ along the same input direction $v_1$, differing only in the output direction. This makes
\begin{equation}\label{eq:sigma-g}
  \Sigma_g:\quad \dot z=\Ac z+v_1w,\qquad y=v_3^\top z,
\end{equation}
the natural auxiliary system, and we write
\begin{equation}\label{eq:gdef}
  g(s):=v_3^\top(sI_n-\Ac)^{-1}v_1,\qquad
  s\notin\spec(\Ac)
\end{equation}
for its transfer function. Neither $\Ac$ nor $v_3$ nor $v_1$ depends on $K$, so $g$ is fixed by the plant and the barrier alone. Closing the loop in Fig.~\ref{fig:interconnection}(b) with $w=-\lambda y$ gives the state matrix $\Ac-\lambda v_1v_3^\top$, which is $A$ at $\lambda=1$ and $\Ac$ at $\lambda=0$.

By Lemma~\ref{lem:circle}, the quantity that must be positive real is not the loop transfer function $\lambda g$ but the return difference $1+\lambda g$ of that loop. Since $g$ is scalar, the matrix determinant lemma \cite[Fact~2.16.3]{bernstein2009matrix} identifies the return difference with a ratio of characteristic polynomials: for any $\mu\in\R$, $1+\mu\,v_3^\top(sI_n-\Ac)^{-1}v_1=p_{\Ac-\mu v_1v_3^\top}(s)/p_{\Ac}(s)$. Taking $\mu=1$ and using $\Ac-v_1v_3^\top=A$ from \eqref{eq:Atilde-rank1} together with
$p_{\Ac}=\varphi$ from Proposition~\ref{prop:Atilde} gives the closed form
\begin{equation}\label{eq:gform}
  g(s)=\frac{p_A(s)}{\varphi(s)}-1
  =-\frac{\sum_{i=1}^{n}(\beta_i-a_i)s^{n-i}}{\varphi(s)},
\end{equation}
where $p_A(s)=s^n+\sum_{i=1}^{n}a_is^{n-i}$; the second equality holds because $p_A$ and $\varphi$ are monic polynomials of the same degree, and it shows that $g$ is strictly proper.

\begin{figure}[t]
  \centering
\tikzset{
  bd/.style={font=\footnotesize, >=stealth', line width=0.5pt},
  blk/.style={draw, line width=0.5pt, rectangle, align=center,
              inner xsep=4pt, inner ysep=3pt, minimum height=6.4mm},
  sum/.style={draw, line width=0.5pt, circle, minimum size=3.4mm, inner sep=0pt},
  dot/.style={circle, fill, minimum size=1.3mm, inner sep=0pt},
}

\begin{tikzpicture}[bd]
  \node[sum]                      (s) at (0,0) {};
  \node[blk, minimum width=9mm]   (K) at (1.30,0) {$-K$};
  \node[blk, minimum width=13mm]  (F) at (3.05,0) {safety\\filter};
  \node[blk, minimum width=24mm]  (P) at (5.65,0) {$\dot x = Ax + Bu$};
  \coordinate (br) at ($(P.east)+(0.45,0)$);
  \coordinate (en) at ($(P.east)+(1.30,0)$);
  \coordinate (fb) at ($(br)+(0,-1.25)$);

  \draw[->] (s)   -- (K);
  \draw[->] (K)   -- node[above,inner sep=1.5pt] {$u_0$} (F);
  \draw[->] (F)   -- node[above,inner sep=1.5pt] {$u^*$} (P);
  \draw     (P.east) -- (br);
  \draw[->] (br)  -- (en);
  \node[above,inner sep=2pt] at ($(P.east)+(0.22,0)$) {$x$};
  \node[dot] at (br) {};

  \draw     (br) -- (fb);
  \draw     (fb) -- (s|-fb);
  \draw[->] (s|-fb) -- (s);
  \node[dot] at (F|-fb) {};
  \draw[->] (F|-fb) -- (F.south);
  \node[above,inner sep=2pt] at ($(F|-fb)!0.45!(s|-fb)$) {$x$};
  \node at ($(s.south)+(-0.20,-0.17)$) {$-$};

  \node[anchor=south west, inner sep=0pt] at ($(s.west)+(-0.15,0.72)$) {(a)};
\end{tikzpicture}

\vspace{2.5mm}

\begin{tikzpicture}[bd]
  \node[sum]                     (s) at (0,0) {};
  \node[blk, minimum width=12mm] (g) at (1.55,0) {$g(s)$};
  \coordinate (br) at ($(g.east)+(0.48,0)$);
  \coordinate (en) at ($(g.east)+(1.30,0)$);
  \coordinate (fb) at ($(br)+(0,-1.15)$);

  \draw[->] (s) -- node[above,inner sep=1.5pt] {$w$} (g);
  \draw     (g.east) -- (br);
  \draw[->] (br) -- (en);
  \node[above,inner sep=2pt] at ($(g.east)+(0.24,0)$) {$y$};
  \node[dot] at (br) {};

  \node[blk, minimum width=8mm] (lam) at ($(fb)!0.52!(s|-fb)$) {$\lambda$};
  \draw     (br) -- (br|-fb);
  \draw     (br|-fb) -- (lam.east);
  \draw     (lam.west) -- (s|-fb);
  \draw[->] (s|-fb) -- (s);
  \node at ($(s.south)+(-0.20,-0.17)$) {$-$};

  \node[anchor=south west, inner sep=0pt] at ($(s.west)+(-0.15,0.62)$) {(b)};
\end{tikzpicture}
\hfill
\begin{tikzpicture}[bd]
  \coordinate (in) at (0,0);
  \coordinate (br) at (0.55,0);
  \node[sum]  (s)  at (2.85,0) {};
  \coordinate (en) at ($(s.east)+(0.85,0)$);
  \coordinate (fb) at ($(br)+(0,-1.15)$);

  \draw[->] (in) -- (s);
  \node[above,inner sep=2pt] at (0.22,0) {$\varrho$};
  \node[dot] at (br) {};
  \draw[->] (s) -- (en);
  \node[above,inner sep=2pt] at ($(s.east)+(0.50,0)$) {$q$};

  \node[blk, minimum width=15mm] (lg) at ($(fb)!0.55!(s|-fb)$) {$-\lambda\,g(s)$};
  \draw     (br) -- (br|-fb);
  \draw[->] (br|-fb) -- (lg.west);
  \draw     (lg.east) -- (s|-fb);
  \draw[->] (s|-fb) -- (s);
  \node[left,inner sep=1.5pt]  at (s.north west) {$+$};
  \node[right,inner sep=1.5pt] at (s.south east) {$-$};

  \node[anchor=south west, inner sep=0pt] at ($(in)+(-0.10,0.62)$) {(c)};
\end{tikzpicture}
  \caption{(a) The safety filter in closed loop, following
  \cite[Fig.~1]{marchese2025hocbf}. (b) The auxiliary system \eqref{eq:sigma-g} under
  $w=-\lambda y$, with state matrix $\Ac-\lambda v_1v_3^\top$. (c) The return difference
  $1+\lambda g$.}
  \label{fig:interconnection}
\end{figure}
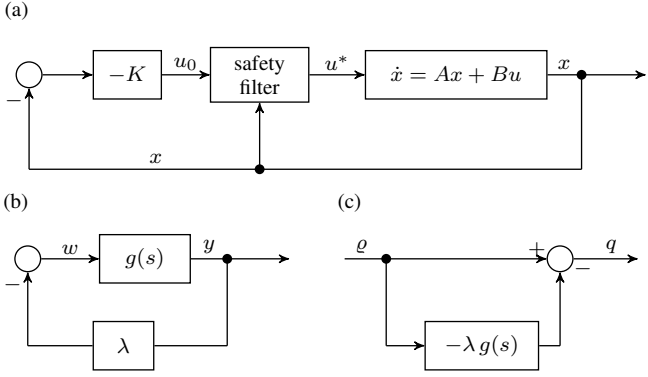

\begin{proposition}\label{prop:hform}
Suppose Assumptions~\ref{as:origin} and \ref{as:reldeg} hold. For fixed
$K\in\R^{m\times n}$ and $G\succ0$, let $A_0$, $\Ac$, $v_1$, and
$v_2$ be defined by \eqref{eq:v1}--\eqref{eq:closedloop}. Then the
loop transfer function and return difference, defined for
$s\notin\spec(\Ac)$ by
\[
g_K(s):=v_2^\top(sI_n-\Ac)^{-1}v_1,\qquad
h_K(s):=1+g_K(s),
\]
satisfy
\begin{equation}\label{eq:hform}
  h_K(s)=1+v_2^\top(sI_n-\Ac)^{-1}v_1
  =\frac{p_{A_0}(s)}{\varphi(s)} .
\end{equation}
Thus, the possible poles of $h_K$ are the points $-\alpha_i$,
and $g_K=h_K-1$ is strictly proper. Moreover, a realisation of $h_K$ is
\begin{equation}\label{eq:sigma-h}
  \Sigma_h:\quad \dot z=\Ac z+v_1\varrho,\qquad q=v_2^\top z+\varrho,
\end{equation}
that is $(A_h,B_h,C_h,D_h)=(\Ac,v_1,v_2^\top,1)$.
\end{proposition}

The plant, the nominal gain and the weighting matrix enter \eqref{eq:hform} only through $p_{A_0}$, and the barrier design only through $\varphi$. The two descriptions agree on the family studied in Section~\ref{subsec:family}: there $v_2=\lambda v_3$, so
$g_K=\lambda g$ and $h_K=1+\lambda g$.

\begin{lemma}\label{lem:minimal-aux}
Suppose Assumption~\ref{as:reldeg} holds. For fixed
$K\in\R^{m\times n}$ and $G\succ0$, let $A_0$, $\Ac$, $v_1$, and
$v_2$ be defined by \eqref{eq:v1}--\eqref{eq:closedloop}. Then the
realisation $(\Ac,v_1,v_2^\top,1)$ of $h_K$ in \eqref{eq:sigma-h} is
minimal if and only if
\begin{equation}\label{eq:nondegenerate}
  p_{A_0}(-\alpha_i)\neq0 \qquad\text{for all }i\in[n],
\end{equation}
that is, if and only if none of the points $-\alpha_i$ is an eigenvalue
of $A_0$.
\end{lemma}

\begin{proof}
In the coordinates $z\mapsto\mathcal{O}z$ the realisation becomes
$(\Gamma,-e_n,v_2^\top\mathcal{O}^{-1},1)$ by \eqref{eq:Ov1}, and
$(\Gamma,e_n)$ is controllable, being in controller canonical form. Observability fails exactly when the numerator $p_{A_0}-\varphi$ of $g_K$ shares a root with $\varphi$, and at a root $-\alpha_i$ of $\varphi$ that numerator equals $p_{A_0}(-\alpha_i)$.
\end{proof}

Condition \eqref{eq:nondegenerate} is generic and is enforced by the choice of the $\alpha_i$; we assume it throughout the remainder of the letter.

\subsection{The Equivalence}

We now connect the three viewpoints needed for the overall stability
analysis: a common quadratic Lyapunov certificate for the two modes,
positive realness of the scalar return difference, and very strict
passivity of its state-space realisation. The equivalence below is the
central result; the subsequent proposition makes the common storage
function explicit, and Theorem~\ref{thm:ges} transfers the certificate
to the state-dependent closed loop.

\begin{theorem}\label{thm:main}
Suppose Assumptions~\ref{as:origin} and \ref{as:reldeg} hold. For fixed
$K\in\R^{m\times n}$ and $G\succ0$, let $A_0$, $\Ac$, $v_1$, and
$v_2$ be defined by \eqref{eq:v1}--\eqref{eq:closedloop}. If
$A_0\neq\Ac$ and condition~\eqref{eq:nondegenerate} holds, then the
following statements are equivalent:
\begin{enumerate}
\item $\Ac$ and $A_0$ admit an SCQLF in the sense of
Definition~\ref{def:cqlf}.
\item The return difference $h_K$ in \eqref{eq:hform} is SSPR in the
sense of Definition~\ref{def:sspr}, equivalently,
      \begin{equation}\label{eq:sspr-cond}
        \Real\!\left\{\frac{p_{A_0}(j\omega)}{\varphi(j\omega)}\right\}>0
        \qquad\text{for all }\omega\in[-\infty,\infty].
      \end{equation}
\item There exists $P=P^\top\succ0$ such that
      \begin{equation}\label{eq:kyp-ours}
        \begin{bmatrix}
          \Ac^\top P+P\Ac & Pv_1-v_2\\[2pt]
          v_1^\top P-v_2^\top & -2
        \end{bmatrix}\prec0 .
      \end{equation}
\end{enumerate}
Whenever these equivalent conditions hold, $A_0$ is Hurwitz and the
system $\Sigma_h$ in \eqref{eq:sigma-h} is VSP.
\end{theorem}

\begin{proof}
\emph{\text{1)} $\Leftrightarrow$ \text{2)}.} $\Ac$ is Hurwitz by Proposition~\ref{prop:Atilde}, $A_0-\Ac$ has rank one, and $(\Ac,v_1,v_2^\top)$ is minimal by Lemma~\ref{lem:minimal-aux}, so Lemma~\ref{lem:circle} with $b=v_1$, $c=v_2$ and $\lambda=1$ gives equivalence with $1+\Real\{g_K(j\omega)\}>0$, which is \eqref{eq:sspr-cond} by Proposition~\ref{prop:hform}. Since $h_K$ is analytic in $\Real\{s\}\geq0$ and tends to one as $s\to\infty$, the strict frequency-domain inequality is equivalent to SSPR.

\emph{\text{2)} $\Leftrightarrow$ \text{3)}.} Apply the Kalman--Yakubovich--Popov lemma \cite[Sec.~3.1.6]{brogliato2020dissipative} to the minimal realisation $(\Ac,v_1,v_2^\top,1)$, for which $D+D^\top=2$.

\emph{Remaining claims.} Lemma~\ref{lem:hurwitz} with $b=v_1$, $c=v_2$ gives $A_0$ Hurwitz. Minimality and Lemma~\ref{lem:sspr-vsp} give VSP.
\end{proof}

\begin{proposition}\label{prop:storage}
Suppose Assumptions~\ref{as:origin} and \ref{as:reldeg} hold. For fixed
$K\in\R^{m\times n}$ and $G\succ0$, let $\Ac$, $A_0$, $v_1$, and
$v_2$ be defined by \eqref{eq:v1}--\eqref{eq:closedloop}. Suppose
$P=P^\top\succ0$, $L\in\R^n$, $W\in\R$, and $\varepsilon>0$ satisfy
\begin{equation}\label{eq:kyp-ours-eq}
  \Ac^\top P+P\Ac=-LL^\top-\varepsilon P,\quad Pv_1=v_2-LW,\quad W^2=2.
\end{equation}
Then $V(z)=z^\top Pz$ is an SCQLF for $\Ac$ and $A_0$, with
\begin{equation}\label{eq:cqlf-explicit}
  A_0^\top P+PA_0=-\varepsilon P-\left(L-Wv_2\right)\left(L-Wv_2\right)^\top .
\end{equation}
\end{proposition}

\begin{proof}
From \eqref{eq:A0-rank1}, $A_0^\top P+PA_0=\Ac^\top P+P\Ac-(v_2v_1^\top P+Pv_1v_2^\top)$. The second equation in \eqref{eq:kyp-ours-eq} gives $Pv_1v_2^\top=v_2v_2^\top-LWv_2^\top$, and adding this to its transpose gives $v_2v_1^\top P+Pv_1v_2^\top=2v_2v_2^\top-LWv_2^\top-v_2WL^\top$. Substituting the first equation in \eqref{eq:kyp-ours-eq} and using $W^2=2$,
\begin{align*}
  A_0^\top P+PA_0
  &= -LL^\top-\varepsilon P-2v_2v_2^\top+W\!\left(Lv_2^\top+v_2L^\top\right)\\
  &= -\varepsilon P-\left(L-Wv_2\right)\left(L-Wv_2\right)^\top\prec0 .
\end{align*}
The inequality $\Ac^\top P+P\Ac=-LL^\top-\varepsilon P\prec0$ is immediate.
\end{proof}

The quadratic storage function certifying very strict passivity of $\Sigma_h$ is therefore, without modification, the common Lyapunov function required by \cite[Thm.~2]{mestres2026dynamical}: with $V(z)=z^\top Pz$, the certified dissipation inequality along \eqref{eq:sigma-h} is $\dot V\leq2\varrho q-\varepsilon V-\norm{L^\top z}^2$, whose margins are the input and output passivity indices of \eqref{eq:vsp}.

\begin{theorem}\label{thm:ges}
Suppose Assumptions~\ref{as:origin} and \ref{as:reldeg} hold. For fixed
$K\in\R^{m\times n}$ and $G\succ0$, let $A_0$ and $\Ac$ be defined by
\eqref{eq:closedloop}. If condition~\eqref{eq:nondegenerate} holds and
$p_{A_0}/\varphi$ is SSPR, then the origin of \eqref{eq:closedloop} is
globally exponentially stable.
\end{theorem}

\begin{proof}
Theorem~\ref{thm:main} supplies $P\succ0$ that is an SCQLF for $\Ac$ and $A_0$ and shows that $A_0$ is Hurwitz. Since $u^*$ is locally Lipschitz, \eqref{eq:closedloop} is a continuous piecewise affine system of the class in \cite[Sec.~III]{pavlov2005convergent} and the two Lyapunov inequalities are the hypotheses of \cite[Thm.~1]{pavlov2005convergent}, giving exponential convergence. By Assumption~\ref{as:origin}, $x(t)\equiv0_n$ is a solution.
\end{proof}

\begin{remark}\label{rem:converse}
Theorem~\ref{thm:main} is an equivalence about \emph{quadratic} certificates. If \eqref{eq:sspr-cond} fails, then no SCQLF exists and, by Lemma~\ref{lem:king}, $A_0\Ac$ has a negative real eigenvalue, so \cite[Thm.~2]{mestres2026dynamical} does not apply. Global exponential stability may nonetheless hold, since the switching in \eqref{eq:closedloop} is state-dependent and confined to $R_\pm$ rather than arbitrary.
\end{remark}

\subsection{A Family of Nominal Gains}\label{subsec:family}

The equivalence above tests any prescribed nominal controller. For
design, it is useful to restrict attention to a one-parameter family
that moves continuously from the inactive-filter controller toward the
open-loop plant and converts the SSPR condition into an explicit scalar
interval.

For design, move along the line through the gain of Proposition~\ref{prop:deactivate}:
\begin{equation}\label{eq:Klambda}
  K(\lambda):=(1-\lambda)\Ld v_3^\top=(1-\lambda)K^\star,\qquad\lambda\in\R .
\end{equation}
Then $v_2^\top=\lambda v_3^\top$ by $\Lambda\Ld=1$, so $\lambda$ is exactly the feedback gain in Fig.~\ref{fig:interconnection}(b) and
\begin{equation}\label{eq:A0-lambda}
  A_0(\lambda)=\Ac-\lambda v_1v_3^\top,
  \quad
  h_\lambda(s)=\frac{(1-\lambda)\varphi(s)+\lambda p_A(s)}{\varphi(s)},
\end{equation}
the second expression following from $h_\lambda=1+\lambda g$ and \eqref{eq:gform}. The family interpolates between the input-output linearising gain at $\lambda=0$, where Proposition~\ref{prop:deactivate} applies, and $K(1)=0$ with $A_0(1)=A$. Multiplying \eqref{eq:gform} by $\varphi(-j\omega)$ gives $\Real\{g(j\omega)\}=N(\omega)/D(\omega)$, where
\begin{equation}\label{eq:PQ}
\begin{gathered}
  N(\omega)=-\!\sum_{m=1}^{n}(-1)^{n-m}\mathfrak{c}_{2m}\,\omega^{2(n-m)},\\
  D(\omega)=\prod\nolimits_{i=1}^{n}\left(\alpha_i^2+\omega^2\right),
\end{gathered}
\end{equation}
where $\mathfrak{c}=\tilde\beta\circledast b$ with $b_i:=\beta_i-a_i$, $b_0=0$, and $\tilde\beta_k:=(-1)^{n-k}\beta_k$. The polynomials $N$ and $D$ have degrees at most $n-1$ and exactly $n$ in $\omega^2$, respectively, so $\Real\{g(j\omega)\}\to0$ as $\omega\to\infty$. Writing $(\cdot)'$ for $\mathrm{d}/\mathrm{d}\omega$, the finite stationary frequencies are the real roots of $N'D-ND'=0$. Hence
\begin{equation}\label{eq:Rbar}
  \bar R:=\sup_{\omega\in\R}\Real\{g(j\omega)\}\geq0,
  \qquad
  \underline R:=\inf_{\omega\in\R}\Real\{g(j\omega)\}\leq0
\end{equation}
are obtained from the values at $\omega=0$, the real roots of
$N'D-ND'=0$, and the limiting value $0$ as $\omega\to\infty$.

\begin{corollary}\label{cor:interval}
Suppose Assumptions~\ref{as:origin} and \ref{as:reldeg} hold and
$p_A(-\alpha_i)\neq0$ for every $i\in[n]$. Let $\bar R$ and
$\underline R$ be defined by \eqref{eq:Rbar}. With the conventions
$-1/\bar R=-\infty$ if $\bar R=0$ and
$-1/\underline R=+\infty$ if $\underline R=0$, the gain
\eqref{eq:Klambda} renders the origin globally exponentially stable
for \eqref{eq:closedloop} whenever
\begin{equation}\label{eq:interval}
  -\frac{1}{\bar R}<\lambda<-\frac{1}{\underline R}.
\end{equation}
For every $\lambda\neq0$ outside this interval, $A_0(\lambda)$ and
$\Ac$ admit no SCQLF.
\end{corollary}

\begin{proof}
By \eqref{eq:A0-lambda}, condition~\eqref{eq:sspr-cond} becomes
$1+\lambda\Real\{g(j\omega)\}>0$ for all $\omega$. For $\lambda>0$
this is equivalent to $\lambda\underline R>-1$, and for $\lambda<0$
to $\lambda\bar R>-1$, which gives \eqref{eq:interval}. For
$\lambda\neq0$, the assumption $p_A(-\alpha_i)\neq0$ implies
$p_{A_0(\lambda)}(-\alpha_i)\neq0$, so Theorems~\ref{thm:main} and
\ref{thm:ges} give the claims. At $\lambda=0$, the stability conclusion
follows from Proposition~\ref{prop:deactivate}.
\end{proof}

\subsection{Feasibility of the Design Inequalities}

The positive-real characterisation determines whether a prescribed
gain admits an SCQLF, whereas \cite[Lem.~3]{mestres2026dynamical}
searches for one gain through two coupled LMIs. To relate these
approaches and determine when that synthesis can succeed at all, we
eliminate the controller variable and reduce feasibility to conditions
on the Lyapunov variable alone.

\begin{equation}\label{eq:hatAB}
\hat A:=A+v_1v_3^\top,\qquad
\hat B:=(I_n+v_1c^\top A^{n-1})B .
\end{equation}
Then $\hat A-\hat BK=\Ac$ for every $K$ \cite[Lem.~3]{mestres2026dynamical}. The design inequalities of \cite[Lem.~3]{mestres2026dynamical} in $Q=Q^\top\succ0$ and $Y\in\R^{m\times n}$ are
\begin{subequations}\label{eq:lmi}
\begin{align}
  AQ+QA^\top+BY+Y^\top B^\top &\prec0,\label{eq:lmi-a}\\
  \hat AQ+Q\hat A^\top+\hat BY+Y^\top\hat B^\top &\prec0,\label{eq:lmi-b}
\end{align}
\end{subequations}
and feasibility yields $P=Q^{-1}$, $K=-YQ^{-1}$ with $x^\top Px$ an SCQLF for $A_0$ and $\Ac$.

\begin{proposition}\label{prop:lmi}
Suppose Assumption~\ref{as:reldeg} holds, and let $\Ac$, $\hat A$, and
$\hat B$ be defined by \eqref{eq:closedloop} and \eqref{eq:hatAB}.
Then the pair \eqref{eq:lmi} is feasible if and only if
there exists a $Q=Q^\top\succ0$ such that
\begin{equation}\label{eq:lmi-b-reduced}
  \Ac Q+Q\Ac^\top\prec0
  \quad\text{and}\quad
  B_\perp^\top\!\left(AQ+QA^\top\right)B_\perp\prec0,
\end{equation}
where the columns of $B_\perp$ span $\ker(B^\top)$.
\end{proposition}

\begin{proof}
Fix $Q\succ0$. Since $Q$ is invertible, every $Y\in\R^{m\times n}$ can be written uniquely as $Y=-KQ$ with $K:=-YQ^{-1}$. Substituting this into the left-hand side of \eqref{eq:lmi-b} and collecting terms gives
$(\hat A-\hat BK)Q+Q(\hat A-\hat BK)^\top=\Ac Q+Q\Ac^\top$, which by Proposition~\ref{prop:Atilde} is independent of $K$ and hence of $Y$. For a given $Q$, therefore, \eqref{eq:lmi-b} holds either for every $Y$ or for none, according to whether $\Ac Q+Q\Ac^\top\prec0$.

Inequality \eqref{eq:lmi-a} is, for a given $Q$, a LMI in $Y$
alone. Write it as $\Psi+U^\top YV+V^\top Y^\top U\prec0$ with $\Psi:=AQ+QA^\top$, $U:=B^\top$ and $V:=I_n$. By the elimination lemma \cite[Sec.~2.6]{boyd1994lmi}, such a $Y$ exists if and only if $W_U^\top\Psi W_U\prec0$ and $W_{V}^\top\Psi W_{V}\prec0$, where the columns of $W_U$ and $W_{V}$ are bases of $\ker(U)$ and $\ker(V)$. Here, $\ker(V)=\{0_n\}$, so the second condition is empty, while $\ker(U)=\ker(B^\top)$, so the first is the second inequality in \eqref{eq:lmi-b-reduced}.

If $(Q,Y)$ satisfies \eqref{eq:lmi}, then $Q$ satisfies both inequalities in \eqref{eq:lmi-b-reduced} by the two paragraphs above. Conversely, if $Q\succ0$ satisfies both, then \eqref{eq:lmi-b} holds for every $Y$, and the elimination lemma shows the existence of a $Y$ satisfying \eqref{eq:lmi-a}.
\end{proof}

Condition \eqref{eq:lmi-b-reduced} involves $Q$ alone, so it can be tested without searching over the gain, and it settles whether the construction of \cite[Lem.~3]{mestres2026dynamical} can produce a nominal controller at all. The set of $Q$ satisfying its first inequality is a nonempty open convex cone, since $\Ac$ is Hurwitz. Further, the change of variables $\bar Q:=\mathcal{O}Q\mathcal{O}^\top$ turns it into $\Gamma\bar Q+\bar Q\Gamma^\top\prec0$, which depends on $\{\alpha_i\}$ alone, and can therefore be characterised once and reused across plants sharing $\mathcal{O}$.

Theorem~\ref{thm:main} tests any $K$, so the admissible set is
\begin{equation}\label{eq:admissible-set}
  \mathcal{K}_\star=\left\{K\in\R^{m\times n}:
  \begin{array}{l}
  p_{A-BK}(-\alpha_i)\neq0\quad\forall i\in[n],\\[-1pt]
  p_{A-BK}/\varphi\ \text{is SSPR}
  \end{array}\right\},
\end{equation}
of which \eqref{eq:interval} describes the intersection with a line and \eqref{eq:lmi} returns a single arbitrary point in the set.

\section{The Flexible Two-Mass Example}\label{sec:example}

We apply the results to \cite[Sec.~III.B]{marchese2025hocbf}, which
\cite[Sec.~V]{mestres2026dynamical} singles out as a case where $\Ac$ is Hurwitz while $A_0\Ac$ has negative real eigenvalues, so that \cite[Thm.~2]{mestres2026dynamical} does not apply.

\subsection{System Data}

Two point masses move in one dimension, coupled by an ideal spring of stiffness $k$ and zero rest length; the force acts on the second mass and the constrained output is the position of the first. With $x=[p_1,p_2,\dot p_1,\dot p_2]^\top$,
\begin{equation}\label{eq:ex-data}
  A=\begin{bmatrix}
    0&0&1&0\\ 0&0&0&1\\
    -\tfrac{k}{m_1}&\tfrac{k}{m_1}&0&0\\
    \tfrac{k}{m_2}&-\tfrac{k}{m_2}&0&0
  \end{bmatrix},\quad
  B=\begin{bmatrix}0\\0\\0\\ \tfrac{1}{m_2}\end{bmatrix},\quad
  c=e_1,
\end{equation}
with $m_1=m_2=\SI{1}{\kilo\gram}$,
$k=\SI{0.1}{\newton\per\metre}$, $d=\SI{0.01}{\metre}$, and $G=1$.
The barrier is $h_0(x)=p_1+d$, so Assumption~\ref{as:origin} holds.
Since $c^\top A^iB=0$ for $i\leq2$ and
$\Lambda=k/(m_1m_2)=\num{0.1}$, the constraint has relative degree
$r=n=4$. The class-$\mathcal{K}$ functions share a slope $\alpha$, so
$\varphi(s)=(s+\alpha)^4$, matching
$K_\alpha=[\alpha^4,4\alpha^3,6\alpha^2,4\alpha]$ in
\cite{marchese2025hocbf}. Finally, $K_{\mathrm{LQR}}$ is the linear
quadratic regulator gain for $Q=\diag(1,0,0,0)$ and $R=1$
\cite{marchese2025hocbf}, with
\begin{equation}\label{eq:ex-Klqr}
    K_{\mathrm{LQR}}=\begin{bmatrix}\num{0.2418}&\num{0.7582}&\num{3.1463}&\num{1.2314}\end{bmatrix} .
\end{equation}

\subsection{The Two Modes}

The plant has a rigid-body mode and an undamped oscillation, so
$p_A(s)=s^4+k(m_1^{-1}+m_2^{-1})s^2=s^4+\num{0.2}s^2$. By Proposition~\ref{prop:Atilde}, $\spec(\Ac)=\{-\alpha\}$ for every $K$ and every $\alpha>0$, which is the analytic counterpart of the observation in \cite{marchese2025hocbf} that the mode with the filter active is stable in its own right; the gain it applies is $K^\star$, equal at $\alpha=4$ to
\begin{equation}\label{eq:ex-Kstar}
  K^\star=\begin{bmatrix}\num{2464.2}&\num{95.8}&\num{2544}&\num{16}\end{bmatrix} .
\end{equation}
Since $\Ac$ has no positive real eigenvalue, the origin is the only equilibrium of the closed loop for every $\alpha$ by \cite[Prop.~2]{mestres2026dynamical}. Under \eqref{eq:ex-Klqr},
\begin{equation}\label{eq:ex-chi0}
  p_{A_0}(s)=s^4+\num{1.2314}s^3+\num{0.9582}s^2+\num{0.4378}s+\num{0.1},
\end{equation}
with roots $-\num{0.4399}\pm\num{0.2418}j$ and $-\num{0.1758}\pm\num{0.6049}j$. Both modes are thus Hurwitz for every $\alpha>0$, and by Theorem~\ref{thm:main} what remains is whether $p_{A_0}/\varphi$ is SSPR.

\subsection{The Positive-Real Test}

Condition~\eqref{eq:sspr-cond} here is the scalar test
\begin{equation}\label{eq:ex-test}
  \Real\!\left\{\frac{p_{A_0}(j\omega)}{(\alpha+j\omega)^4}\right\}>0
  \qquad\text{for all }\omega\in\R
\end{equation}
in the single design parameter $\alpha$. Evaluating it on a frequency grid and bisecting on $\alpha$ gives a critical value
\begin{equation}\label{eq:ex-alphastar}
  \alpha^\star=\num{0.9138},
\end{equation}
below which \eqref{eq:ex-test} holds, with the violation first appearing at
$\omega=\num{0.960}$. Fig.~\ref{fig:nyquist} shows the locus of $h(j\omega)$ for several $\alpha$ and the margin $\min_\omega\Real\{h(j\omega)\}$ against $\alpha$. The value \eqref{eq:ex-alphastar} agrees with the numerical finding in \cite[Sec.~IV]{marchese2025hocbf} that an SCQLF exists for $0<\alpha<1$. Computing the eigenvalues of $A_0\Ac$ independently, a negative real eigenvalue appears at the same threshold, as Lemma~\ref{lem:king} requires.

\begin{figure}[t]
  \centering
  \includegraphics[width=\columnwidth]{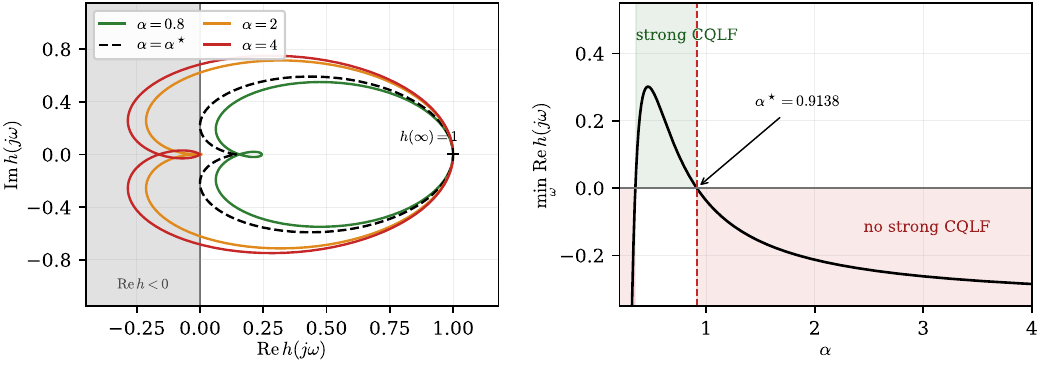}
  \makebox[\columnwidth]{\makebox[.5\columnwidth]{(a)}\makebox[.5\columnwidth]{(b)}}
  \caption{(a) Locus of $h(j\omega)=p_{A_0}(j\omega)/\varphi(j\omega)$ under $K_{\mathrm{LQR}}$; SSPR fails when the locus enters the shaded half-plane. (b) The margin $\min_\omega\Real\{h(j\omega)\}$ against $\alpha$, crossing zero at $\alpha^\star=\num{0.9138}$.}
  \label{fig:nyquist}
\end{figure}

\subsection{What the Threshold Predicts}

Simulating \eqref{eq:closedloop} from $x_0=[5,5,0,0]^\top$ gives the three regimes in Fig.~\ref{fig:trajectories}. At $\alpha=\num{0.8}$ the test holds, Theorem~\ref{thm:ges} applies, and $p_1$ converges after a single activation of the filter. At $\alpha=2$ the test fails, so no SCQLF exists, yet the trajectory still converges with persistent switching and a slowly decaying oscillation. At $\alpha=4$ it diverges, and the divergence threshold we observe near $\alpha\approx\num{3.25}$ matches \cite[Sec.~IV]{marchese2025hocbf}.

The interval between $\alpha^\star$ and $\alpha\approx\num{3.25}$ is the gap described in Remark~\ref{rem:converse}: loss of the quadratic certificate is necessary for instability here but not sufficient, because the switching is state-dependent and confined to $R_\pm$.

\subsection{Redesigning the Nominal Gain}

Suppose $\alpha=4$ is imposed by the barrier specification. Computing \eqref{eq:PQ} and \eqref{eq:Rbar} gives $\bar R=0$ and $\underline R=-\num{1.3443}$, so Corollary~\ref{cor:interval} yields
\begin{equation}\label{eq:ex-interval}
  \lambda\in\left(-\infty,\ \num{0.7439}\right).
\end{equation}
Taking $\lambda=\num{0.5}$,
\begin{equation}\label{eq:ex-K}
  K(\num{0.5})=\tfrac12K^\star
  =\begin{bmatrix}\num{1232.1}&\num{47.9}&\num{1272}&\num{8}\end{bmatrix},
\end{equation}
for which $\min_\omega\Real\{h_\lambda(j\omega)\}=\num{0.328}$ and the origin is globally exponentially stable; Fig.~\ref{fig:trajectories}(b) compares the two gains at this $\alpha$.

\begin{figure}[t]
  \centering
  \includegraphics[width=\columnwidth]{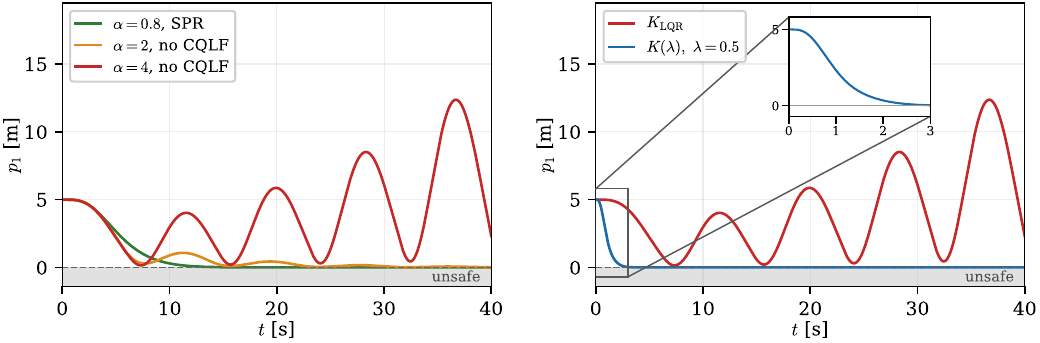}
  \makebox[\columnwidth]{\makebox[.5\columnwidth]{(a)}\makebox[.5\columnwidth]{(b)}}
  \caption{Closed-loop position $p_1$ from $x_0=[5,5,0,0]^\top$ with $d=\num{0.01}$. (a) The three regimes under $K_{\mathrm{LQR}}$. (b) At $\alpha=4$, the trajectories under $K_{\mathrm{LQR}}$ and \eqref{eq:ex-K}. All trajectories remain in the safe set, shown unshaded.}
  \label{fig:trajectories}
\end{figure}

\section{Conclusions}\label{sec:conclusion}

This paper presented a design framework for safe and globally
exponentially stabilising controllers for linear systems with a single
full-relative-degree affine constraint. We showed that the
filtered-mode spectrum is prescribed entirely by the barrier tuning
and derived an explicit nominal controller that satisfies the barrier
condition everywhere, thereby eliminating switching. For a prescribed
nominal controller, we proved that the nominal and filtered modes admit
an SCQLF if and only if the ratio of the nominal characteristic
polynomial to the barrier polynomial is SSPR. This equivalence reduces
certification to a scalar frequency-domain test and yields an explicit
interval of admissible gains.

Building on the conference results, the extended analysis established
a passivity interpretation of the positive-real condition. The
quadratic storage function obtained from the KYP factorisation is,
without modification, an SCQLF for the two modes. We also showed that
full relative degree forces the input matrix to have rank one and
reduced the existing LMI synthesis to a feasibility condition in a
single matrix variable. These results clarify both the structure of
the safety-filtered dynamics and the relationship between the
frequency-domain and LMI design approaches.

The flexible two-mass example showed that loss of the quadratic
certificate can precede instability and that the proposed redesign
restores a safe and globally exponentially stable closed loop. Thus,
the positive-real condition is exact for common quadratic
certificates, although its failure does not imply instability under
state-dependent switching. Future work will address this gap,
intermediate relative degrees, multiple affine constraints, and
nonlinear systems.

\section*{Acknowledgment}
The authors would like to thank Pol Mestres for his insightful
discussions on this topic.

\bibliography{ref}
\bibliographystyle{IEEEtran}

\end{document}